\documentclass[reqno,11pt]{article} 
\newcommand{\version}{1-09-2017}

\usepackage{amsmath,amssymb,amsthm,cite,slashed,mathtools,subcaption}
\usepackage{color,graphicx}
\usepackage{tikz} 
\usetikzlibrary{arrows.meta, calc}
\usepackage[margin=3cm]{geometry} 

\title{String chopping and time-ordered products of linear
string-localized quantum fields} 

\author{
Lucas T. Cardoso,\thanks{Supported by CAPES.}\quad
Jens Mund\thanks{Supported by CNPq.
}\\[6pt]
{\small Departamento de F\'isica,
Universidade Federal de Juiz de Fora,}\\
{\small 36036--900 Juiz de Fora, MG, Brasil}\\[12pt]
Joseph C.~V\'arilly
\\[6pt]
{\small Escuela de Matem\'atica, Universidad de Costa Rica,
San Jos\'e 11501, Costa Rica}\\[3pt]
}

\date{\version \\[3ex]} 

\newcounter{P}         
\renewcommand{\theP}{(P\arabic{P})} 

\DeclareMathOperator{\linspan}{span} 
\DeclareMathOperator{\sgn}{sgn}     
\DeclareMathOperator{\T}{T}         

\newcommand{\RR}{\mathbb{R}}        

\newcommand{\later}{\succcurlyeq}   
\newcommand{\earlier}{\preccurlyeq} 

\newtheorem{prop}{Proposition}[section]
\newtheorem{lema}{Lemma}[section]

\theoremstyle{definition}
\newtheorem{defn}{Definition}[section]

\theoremstyle{remark}
\newtheorem{remk}{Remark}[section]
\newtheorem*{claim}{Claim}

\numberwithin{equation}{section} 

\newcommand{\scalprod}[2]{(#1\mathbin{,}#2)} 
\newcommand{\Scalprod}[2]{\bigl(#1\mathbin{,}#2\bigr)} 

\newcommand{\calD}{\mathcal{D}}
\newcommand{\calG}{\mathcal{G}}     

\newcommand{\calL}{\mathcal{L}}
\newcommand{\calM}{\mathcal{M}}     

\newcommand{\eps}{\varepsilon}
\newcommand{\half}{\tfrac{1}{2}}

\newcommand{\w}{\wedge}               
\newcommand{\x}{\times}               
\newcommand{\xx}{\boldsymbol{x}}      
\newcommand{\xyx}{\times\cdots\times} 

\newcommand{\hideqed}{\renewcommand{\qed}{}} 
\newcommand{\set}[1]{\{\,#1\,\}}      
\newcommand{\word}[1]{\quad\text{#1}\quad} 

\newcommand{\Spd}{H}                

\renewcommand{\lor}{\Lambda}        

\newcommand{\pt}{{\mathrm{p}}}

\newcommand{\field}{\varphi}        
\newcommand{\fieldPt}{\field_\pt}   
\newcommand{\fieldSt}{\field} 

\newcommand{\APt}{A^\pt}            
\newcommand{\ASt}{A}                

\newcommand{\Wickli}{\mathopen{:}}
\newcommand{\Wickre}{\mathclose{:}}

\newcommand{\Sfin}{S^\mathrm{fin}}  
\newcommand{\Sinfty}{S^\infty}      
\newcommand{\String}[1]{S_{#1}}     

\newcommand{\Eext}{E_{\mathrm{ext}}}  
\newcommand{\Eint}{E_{\mathrm{int}}}  
\newcommand{\LargeDiag}{\Delta_n}     
\newcommand{\Vext}{I_{\mathrm{ext}}}  

\begin{document} 

\maketitle 

\begin{abstract}
For a renormalizability proof of perturbative models in the
Epstein--Glaser scheme with string-localized quantum fields, one needs
to know what freedom one has in the definition of time-ordered
products of the interaction Lagrangian. This paper provides a first
step in that direction.

The basic issue is the presence of an open set of $n$-tuples of
strings which cannot be chronologically ordered. We resolve it by
showing that almost all such string configurations can be dissected
into finitely many pieces which can indeed be chronologically ordered.
This fixes the time-ordered products of linear field factors outside a
nullset of string configurations. (The extension across the nullset,
as well as the definition of time-ordered products of Wick monomials,
will be discussed elsewhere.)
\end{abstract}

\section{Introduction} 
\label{sec:Intro}

The three pillars of relativistic quantum field theory (QFT) are
positivity of states, positivity of the energy and locality of
observables (or Einstein causality). Any attempt to reconcile them
leads to the well-known singular behaviour of quantum fields at short
distances (UV singularities)~\cite{Strocchi93}, which becomes worse
with increasing spin. This rules out the direct construction of
interacting models for particles with spin/helicity $s \geq 1$ in a
frame which incorporates the three principles from the beginning.

The usual way out is gauge theory (GT), where one relaxes the
principle of positivity of states in a first step, and divides out the
unphysical degrees of freedom (negative norm states and ghost fields)
at the end of the construction. This approach has been extremely
successful and is the basis of the Standard Model of elementary
particle physics. However, it has some shortcomings: the intermediate
use of unphysical degrees of freedom does not comply well with
Ockham's razor; the approach does not provide a direct construction of
charge-carrying physical fields; it excludes an energy-momentum tensor
for massless higher helicity particles~\cite{WeinbergWitten}. Finally,
many features of models must be put in by hand instead of being
explained, like for example the shape of the Higgs potential, and
chirality of the weak interactions.

There is an alternative, relatively recent but conservative approach
\cite{MSY,MundOliveira,Bert_BeyondGauge,PlaschkeYngvason}, which keeps
positivity of states and instead relaxes the localization properties
of (unobservable) quantum fields: These fields are not point-local,
but instead are localized on Mandelstam strings extending to spacelike
infinity~\cite{MSY,MSY2}. Such a string, not to be confused with the
strings of string theory, is a ray emanating from an event $x$ in
Minkowski space in a spacelike%
\footnote{\label{fn:lightlike}%
The choice of \emph{space}-like strings is motivated by the known fact
that in every massive model charge-carrying field operators are
localizable in spacelike cones~\cite{BuF}. It seems, however, that our
constructions go through also for \emph{light}-like strings, replacing
$\Spd$ by the forward light cone.}
direction $e$,
\begin{equation}
\label{eqSxe} 
\String{x,e} \doteq x + \RR_0^+ e. 
\end{equation}
Our quantum fields are operator-valued distributions $\field(x,e)$, 
where $x$ is in Minkowski space and $e$ lies in the manifold of
spacelike directions 
\begin{equation}
\label{eqSpd} 
\Spd \doteq \set{ e \in \RR^4 : e \cdot e = -1}.
\end{equation}  
The field $\fieldSt(x,e)$ is localized on the string
$\String{x,e}$ in the sense of compatibility of quantum observables:
If the strings $\String{x,e}$ and $\String{x',e'}$ are spacelike
separated,%
\footnote{Indeed, the distributional character of the fields requires
that $\String{x',e''}$ be spacelike separated from $\String{x,e}$
for all $e''$ in an open neighborhood of $e'$.}
then 
\begin{equation} 
\label{eqFieldLoc} 
[\field(x,e), \field(x',e')] = 0.
\end{equation}
It has been shown~\cite{BuF} that in the massive case this is the
worst possible ``non-locality'' for unobservable fields which is
consistent with the three mentioned principles (in particular with
locality of the observables), and that this weak type of localization
still permits the construction of scattering states.
Free string-localized fields for any spin with good UV behaviour have
been constructed in a Hilbert space without ghosts. Among these are
string-localized fields which differ from their badly behaved
point-localized counterparts by a gradient
\cite{MundOliveira,PlaschkeYngvason}. They allow for
the construction of string-localized energy-momentum tensors for any
helicity~\cite{MundRehrenSchroer_letter}, evading the Weinberg--Witten
theorem~\cite{WeinbergWitten}. 
In the (perturbative) construction of interacting models, one uses an
interaction Lagrangian which differs from a point-localized
counterpart by a divergence. Then the classical action is the same for
both Lagrangians. (This is analogous to gauge theory, where two
Lagrangians in different gauges yield the same action.) 

The requirement that this equivalence survive at the quantum level
leads to renormalization conditions which we call \emph{string
independence (SI)} conditions, reminiscent of the Ward identities in
gauge theory. They are quite restrictive: In particular, they imply
features like chirality of weak interactions~\cite{Rosalind}, the
shape of the Higgs potential~\cite{MundSchroer_AbHiggs} and the Lie
algebra structure in models with self-interacting vector
bosons~\cite{Bert_BeyondGauge}. It is not clear at the moment if this
approach leads to the same models as the gauge-theoretic one.

A proof of renormalizability at all orders in this approach 
is missing, up to now. 
The present paper is meant as a first step in this direction. We aim
at the perturbative construction of interacting models within the
Epstein--Glaser scheme~\cite{EG}. This approach is based on the Dyson
series expansion of the $S$-matrix in terms of time-ordered products
of the interaction Lagrangian, which is a Wick polynomial in the free
fields. In the case of point-localized fields, renormalizability
enters as follows. The time-ordered products of $n$ Wick monomials
$W_i$ are basically characterized by symmetry and the factorization
property, namely
\begin{equation} 
\label{eqTPt} 
\T W_1(x_1)\cdots W_n(x_n) 
= \T W_1(x_1)\cdots W_k(x_k) \, \T W_{k+1}(x_{k+1})\cdots W_n(x_n) 
\end{equation}
whenever each event in $\{x_1,\dots,x_k\}$ is ``later'' than each
event in $\{x_{k+1},\dots,x_n\}$. (We say that 
\emph{$x$ is later than~$y$} if there is a reference frame such that
$x^0 > y^0$.) Indeed, these properties recursively fix the $T$-products
off the thin diagonal 
$\set{(x_1,\dots,x_n) \in \RR^{4n} : x_1 =\cdots= x_n}$ (or
equivalently, by translation invariance, outside the origin
of~$\RR^{4(n-1)}$): see~\cite{EG,BruFred00}. In this $x$-space
approach, the ``UV~problem'' of divergences consists in the extension
across the origin, which is not unique: At every order $n$ one has a
certain number of free parameters. If the short-distance scaling
dimension of the interaction Lagrangian is not larger than~$4$, then
this number does not increase with the order, and one can fix all free
parameters by a finite set of normalization conditions: the model is
renormalizable~\cite{EG}, see~\cite{CardosoRenormPoint} for a review
of the argument.

In the present paper, we initiate the corresponding discussion for
string-localized quantum fields $\field(x,e) $ by considering
time-ordered products of \emph{linear} fields 
$\T\field(x_1,e_1)\cdots\field(x_n,e_n)$. 
(The case of Wick monomials of order $>1$ is left for a future
investigation.) These are required to be symmetric and to satisfy the
factorization property, namely Eq.~\eqref{eqTPt} must hold, with
$W_i(x_i)$ replaced by $\fieldSt(x_i,e_i)$, whenever each of the first
$k$ strings is later than each of the last $n-k$ strings.

The basic problem, already present at order~$2$, is that two strings
generically are not comparable in the sense of time-ordering. In fact,
there is an open set of pairs $(x,e),(x',e')$ corresponding to strings
which are not comparable, see Lemma~\ref{R>S}. Thus the $T$-product of
two fields is undefined on an open set, which leaves an infinity of
possible definitions instead of finitely many parameters already at
second order, jeopardizing renormalizability. For three and more
strings, the problem becomes worse, see Fig.~\ref{fg:big-three} for a
typical example.

\begin{figure}[htb]
\centering
\begin{tikzpicture}[>=Stealth]
\coordinate (A) at (-1,0) ; 
\coordinate (B) at (2,0) ;
\coordinate (C) at (1,2.2) ;
\coordinate (X) at  ($ (A)!-0.3!(B) $) ;
\coordinate (Xa) at ($ (A)!-0.05!(B) $) ;
\coordinate (Xb) at ($ (A)!0.05!(B) $) ;
\coordinate (Xl) at ($ (A)!0.6!(B) $) ;
\coordinate (Xe) at ($ (A)!1.3!(B) $) ;
\coordinate (Y) at  ($ (B)!-0.3!(C) $) ;
\coordinate (Ya) at ($ (B)!-0.05!(C) $) ;
\coordinate (Yb) at ($ (B)!0.05!(C) $) ;
\coordinate (Yl) at ($ (B)!0.5!(C) $) ;
\coordinate (Ye) at ($ (B)!1.3!(C) $) ;
\coordinate (Z) at  ($ (C)!-0.3!(A) $) ;
\coordinate (Za) at ($ (C)!-0.05!(A) $) ;
\coordinate (Zb) at ($ (C)!0.05!(A) $) ;
\coordinate (Zl) at ($ (C)!0.6!(A) $) ;
\coordinate (Ze) at ($ (C)!1.3!(A) $) ;
\draw (X) -- (Xa)  (Y) -- (Ya)  (Z) -- (Za) ;
\draw[->] (Xb) -- (Xl) node[above=3pt] {$S_1$} -- (Xe) ;
\draw[->] (Yb) -- (Yl) node[above right] {$S_2$} -- (Ye) ;
\draw[->] (Zb) -- (Zl) node[above left] {$S_3$} -- (Ze) ;
\foreach \pt in {X,Y,Z} \fill (\pt) circle (1.5pt) ;
\end{tikzpicture} 
\caption{Three strings, none of which is later than the other two
(in three-dimensional spacetime -- 
the time arrow points out of the viewing plane).}
\label{fg:big-three} 
\end{figure}
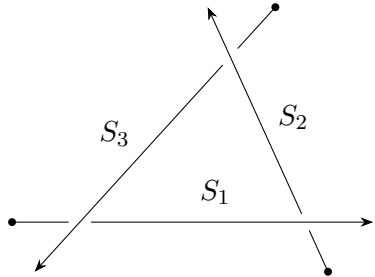

To overcome this problem, we prove first that $n$ strings which do not
touch each other can be chopped up into finitely many pieces which are
mutually comparable. This is our main result. It is shown first for
$n = 2$ in a constructive way (Prop.~\ref{StringChopping2}), and then
for $n > 2$ with a non-constructive proof
(Prop.~\ref{StringChopping}).

We then proceed to show how this purely geometric result fixes the
time-ordered products $\T\field(x_1,e_1)\cdots\field(x_n,e_n)$ outside
the nullset $\LargeDiag$ of strings that touch each other. In
particular, they turn out to satisfy Wick's expansion. Again, this is
first shown for $n = 2$ (Prop.~\ref{T2Diag}), where the product
$\T\field\field$ is fixed by its vacuum expectation value (the Feynman
propagator), and then for $n > 2$ (Prop.~\ref{TWickLinearDiag}).

In the extension of the $T$-products across $\LargeDiag$, the scaling
degrees~\cite{SchulzPhD} of the Feynman propagator with respect to 
the various submanifolds of~$\Delta_2$ 
(in view of Wick's theorem) have to be compared with the respective
codimensions. We give an example in Appendix~\ref{app:IntegrateT}, but
leave the general discussion open for future publications.

The article is organized as follows. Section~\ref{sec:Geometry} is
concerned with geometry: We define the time-ordering prescription for
strings, i.e., the ``later'' relation, and prove our main geometrical
result on the chopping of strings. In
Section~\ref{sec:TimeOrderedProd}, we state the axioms for the
time-ordered product of string-localized fields, and show that (in the
case of linear factors) it is fixed outside the set $\LargeDiag$ and
satisfies Wick's expansion. In Section~\ref{sec:Outlook}, we comment
on a problem that arises in extending the present results to Wick
monomials of order~$> 1$.

We close the introduction with some further details. 
Our fields are covariant under a unitary representation $U$
of the proper orthochronous Poincar\'e group: 
\begin{equation}  
\label{eqCovTens} 
U(a,\Lambda)\, \field(x,e) \, U(a,\Lambda)^{-1}
= \field(a + \Lambda x, \Lambda e), 
\end{equation}
where $a \in \RR^4$ is a translation and $\lor$ is a Lorentz
transformation. 
(This is the scalar case, which we consider here for sake of
notational convenience. The fields may have vector or tensor indices
which also transform, see~\cite{MundOliveira}.)
The irreducible sub-representations of $U$ correspond to the particle
types described by $\field$.%
\footnote{Such fields exist for any spin/helicity
\cite{MSY2,PlaschkeYngvason}.}
We consider here only the case of bosons, and we exclude explicitly
the case of Wigner's massless ``infinite spin''
particles~\cite{Wig48}. It has been shown in~\cite{MSY2} that then our
string-localized free massive field $\field(x,e)$ is of the form
\begin{equation} 
\label{eqFieldSt} 
\field(x,e) = \int_0^\infty ds\, u(s)\, \fieldPt(x + se),
\end{equation}
where $\fieldPt$ is some point-localized free field, and $u$ is some
real-valued function with support in the positive reals.

Of course, one might define the time-ordered product
$\T\field(x,e)\field(x',e')$ by first taking the point-local one and
then integrating, as in Eq.~\eqref{eqTInt}. (Our Props.~\ref{T2Diag}
and~\ref{TWickLinearDiag} may be obtained this way.) However, when it
comes to renormalization (or extension), this procedure misses the
central point of our approach: The point-local Feynman propagator for
higher spin fields (or derivatives of scalar fields) is not unique due
to its bad UV behaviour, and leaves the freedom of adding delta
function renormalizations. This freedom is not undone by the
subsequent integrations. On the other hand, the UV behaviour of
$\field$ is better than that of the point-local field $\fieldPt$ just
due to the integration~\cite{MundOliveira}, and therefore in general
the $T$ product has less freedom. We give an example in
Appendix~\ref{app:IntegrateT}. We conclude that it is worthwhile to
take the string-localized $\field$ seriously as the basic building
block (and not to overburden the $T$-product by continuity assumptions
permitting exchange of integration and time ordering).

\section{Geometric time-ordering} 
\label{sec:Geometry} 

In a given Lorentz frame $\{e_{(0)},e_{(1)},e_{(2)},e_{(3)}\}$, the time
coordinate of an event $x$ is just $x \cdot e_{(0)}$, and an event $x$
occurs ``later'' than an event $y$ in this frame if
$(x - y) \cdot e_{(0)} > 0$.%
\footnote{We adopt the convention that the metric has signature
$(1,-1,-1,-1)$.}
We therefore say that $x$ is later than $y$ if there is some timelike
future-pointing vector $u$ such that $(x - y) \cdot u > 0$. As a
direct consequence of Lemma~\ref{VSelfdual}, this is equivalent to the
condition that $x$ be outside the closed backward light cone of~$y$.
(For general geometric definitions and conventions that will be used
throughout the rest of this work, see Appendix~\ref{app:BasicGeo}.) We
take this as a definition.

\begin{defn}[Posteriority relation] 
\label{df:Posteriority}
For $x,y \in \RR^4$ we say that $x$ is \emph{later than} $y$, in
symbols $x \later y$, if $x$ is not contained in the past light cone
of~$y$:
\begin{equation} 
\label{eqTOGeo} 
x \later y  :\Longleftrightarrow  x \notin \overline{V_-(y)}.
\end{equation}
For subsets $R,S \subset \RR^4$, we say that $R$ is later than $S$,
symbolically $R \later S$, if all points in $R$ are later than all
points in $S$. If either $R \later S$ or $S \later R$, we say $R$ and
$S$ are \emph{comparable}; otherwise, we say $R$ and $S$ are
\emph{incomparable}.
\end{defn}

(For a point $x \in \calM$ and for $R \subset \calM$ we write simply
$x \later R$ instead of $\{x\} \later R$.)

Two warnings are in order.
Physically, the \emph{posteriority} relation $x \later y$ must be
distinguished from the \emph{causality} relation 
$x \in \overline{V_+(y)}$,
which means that $y$ can influence the event $x$ either by way of the
propagation of some material phenomenon or some electromagnetic
effect. Mathematically, posteriority is not an order relation: 
it is not a total order, since not every pair of subsets is
comparable; nor is it even a partial order, since
it is not transitive. Fig.~\ref{fg:big-three} shows a counterexample
to transitivity in $n \geq 3$ dimensions: there holds 
$S_1 \later S_2$ and $S_2 \later S_3$, but \emph{not} 
$S_1 \later S_3$. (Thus, if we write $S_1 \later S_2 \later S_3$ in
the sequel, then this means only that the first two of these relations
hold.)

\subsection{Generalities on the posteriority relation} 
\label{ssc:post-hoc}

We establish some properties of this time-ordering relation which are
relevant for the proof of Propositions~\ref{StringChopping2} and
\ref{StringChopping}. First, note that for two regions $R,S$ in
Minkowski space we have
\begin{equation}
\label{eqR>S} 
R \later S  \iff  R \cap \overline{V_-(S)} = \emptyset.
\end{equation}

\begin{lema} 
\label{R>S} 
Let $R,S \subset \RR^4$. 
\begin{enumerate}
\item[\textup{(i)}]
Both $R \later S$ and $S \later R$ hold if and only if $R$ and $S$ are
spacelike separated.
\item[\textup{(ii)}]
$R$ and $S$ are incomparable if and only if both
$R \cap \overline{V_-(S)} \neq \emptyset$ and
$R \cap \overline{V_+(S)} \neq \emptyset$.
\end{enumerate}
\end{lema}

\begin{proof}
Note first that  $R \cap \overline{V_-(S)} = \emptyset$
is equivalent to $S \cap \overline{V_+(R)} = \emptyset$.
Thus, the condition $(S \later R) \w (R \later S)$ is equivalent, by
Eq.~\eqref{eqR>S}, to $S \cap \overline{V_-(R)} = \emptyset$
and $S \cap \overline{V_+(R)} = \emptyset$. But this is
$S \cap \bigl( \overline{V_+(R)}\cup\overline{V_-(R)} \bigr)
= \emptyset$, which means just that $S$ is spacelike separated
from~$R$. This proves~(i). 
Item~(ii) is a direct consequence of Eq.~\eqref{eqR>S}. 
\end{proof}

\begin{lema} 
\label{Hyperplane} 
Let $\Sigma$ be a spacelike hyperplane of the form 
$\Sigma = a + u^\perp$, where $u$ is a future-pointing timelike vector
and $a \in \Sigma$. Any event $x \in \RR^4$ satisfies
$x \later \Sigma$ iff $(x - a) \cdot u > 0$, that is, $x$
is ``above''~$\Sigma$.
\end{lema}

\begin{proof}
Firstly, note that the condition $x \later \Sigma$ means, by
definition, that $\forall y \in \Sigma$, 
$x - y \notin \overline{V_-}$.
Moreover, $\forall y \in \Sigma$ there holds $u \cdot (a - y) = 0$,
and consequently
$$
u \cdot (x - y) = u \cdot (x - a).
$$
Now suppose $u \cdot (x-a) > 0$, and let $y \in \Sigma$, then
$u \cdot (x - y) \equiv u \cdot (x - a) > 0$, which implies 
$x - y \notin \overline{V_-}$. This shows that $x \later \Sigma$. We
prove inverse direction contrapositively. Suppose, \emph{ad absurdum},
that $u \cdot (x - a)\leq 0$. If $u \cdot (x - a) = 0$, then
$x \in \Sigma$ and thus $\neg(x \later \Sigma)$. On the other hand,
$u \cdot (x - a) < 0$ implies as above that $x \earlier \Sigma$. But
then $x$ cannot be later than $\Sigma$ by Lemma~\ref{R>S}, since the
causal complement of $\Sigma$ is empty. In both cases, we have
$\neg(x \later \Sigma)$. Summarizing, we have shown that
$x \later \Sigma$ is equivalent to $(x - a) \cdot u > 0$.
\end{proof}

The (motivating) characterization of the relation $x \later y$, namely
the condition that $x^0 > y^0$ in some reference frame, can now be
written as the condition that there exists a spacelike hyperplane
$\Sigma$ that is separating in the sense that
$x \later \Sigma \later y$. As we have seen, this is equivalent to our
Definition~\ref{df:Posteriority}. 
The same holds for finite string segments.

However, for infinitely extended strings, the existence of a spacelike
separating hyperplane is a sufficient but not necessary condition for
posteriority as defined in Def.~\ref{df:Posteriority}. (A
sufficient and necessary condition would be the existence of a
spacelike \emph{or lightlike} separating hyperplane. But we don't need
this statement and therefore refrain from proving it.)

\begin{lema} 
\label{SSHyperplane}
Let $R_1 \subseteq S_1$, $R_2 \subseteq S_2$ be two connected subsets
of the strings $S_1$ and $S_2$ (which comprises points, finite string
segments and the entire string, exhausting all possibilities). If
there is a spacelike hyperplane $\Sigma$ such that
$R_1 \later \Sigma \later R_2$, then $R_1 \later R_2$.
\end{lema}

\begin{proof}
Let $\Sigma \doteq a + u^\perp$ satisfy the hypothesis, and let 
$z_1 \in R_1$ and $z_2 \in R_2$. Then by Lemma \ref{Hyperplane} there
holds $(z_1 - a) \cdot u > 0$ and $(a - z_2) \cdot u > 0$. Adding
these two inequalities yields $(z_1 - z_2) \cdot u > 0$, and since
$u \in V_+$ we must have $z_1 - z_2 \notin \overline{V_-}$ by
Lemma~\ref{VSelfdual} (with reversed signs), that is $z_1 \later z_2$.
This completes the proof.
\end{proof}

\paragraph{Comparability.}
For points, the posteriority relation is linear insofar as any
pair of distinct events $x \neq y \in \RR^4$ is comparable, i.e.,
either $x \later y$ or $y \later x$ holds. The first problem we
encounter in the definition of time-ordered products is that this is
not so for disjoint strings. It may happen that one string enters into
the past \emph{and} into the future of another one, and in this case
(only) the two strings are not comparable, by
Lemma~\ref{R>S}(ii). On the other hand, an event and a string are
always comparable whenever they are disjoint.

\begin{lema} 
\label{comparable} 
Let $S$ be a string and $x \in \RR^4 \setminus S$ an event disjoint
from~$S$. Then either $S \later \{x\}$ or $\{x\} \later S$.
\end{lema} 

\begin{proof}
Suppose that neither $S \later x$ nor $x \later S$ holds. Then, by
Eq.~\eqref{eqR>S}, both $S \cap \overline{V_-(x)}$ and
$\{x\} \cap \overline{V_-(S)}$ are nonempty. However, since
$S \cap \overline{V_-(x)} \neq \emptyset$ if and only if
$\{x\} \cap \overline{V_+(S)} \neq \emptyset$, this would entail
$x \in \overline{V_-(S)} \cap \overline{V_+(S)} = S$. The result
follows.
\end{proof}

\paragraph{Transitivity.} 
Time-ordering of events is not transitive. But it has a similar
property, which we might call ``weak transitivity'': if
$y_1 \later y_2$ and $x \not\later y_2$, then $y_1 \later x$. This
fact is the basis for the proof that Bogoliubov's S-matrix satisfies
the functional equation~\cite{EG}, 
that in turn implies locality of the interacting fields in the
Epstein--Glaser construction~\cite{EG}.
Again, this does not hold for strings -- which is why a
string-localized interaction may lead in general to completely non-local
interacting fields. An example is illustrated in
Fig.~\ref{fg:big-three}, which shows three strings satisfying
$S_1 \later S_2$ and $S_3 \not\later S_2$, however $S_1$ is not later
than $S_3$.

On the other hand, weak transitivity does hold for two strings with
respect to an event.

\begin{lema} 
\label{WeakTrans} 
Let two strings $S_1,S_2$ and an event
$x \in \RR^4 \setminus (S_1 \cup S_2)$ be such that
$$
S_1 \later S_2  \word{and}  \{x\} \not\later S_2.  
$$
Then $S_1 \later \{x\}$. 
\end{lema}

\begin{proof}
By Eq.~\eqref{eqR>S}, the premise $S_1 \later S_2$ may be written as
$S_1 \cap \overline{V_-(S_2)} = \emptyset$. Also, $x \not\later S_2$
means $x \in \overline{V_-(S_2)}$ and consequently
$\overline{V_-(x)} \subset \overline{V_-(S_2)}$. Thus, we must have
$S_1 \cap \overline{V_-(x)} = \emptyset$, and again by 
Eq.~\eqref{eqR>S}, we get $S_1 \later \{x\}$.  
\end{proof}

\begin{defn} 
\label{df:latest-member}
Given $n$ subsets $R_1,\dots,R_n$ of Minkowski space, we say that
$R_1$ is a \emph{latest member} of the set $\{R_1,\dots,R_n\}$ if
$R_1 \later R_i$ for each $i = 2,\dots,n$. 
\end{defn}

Here the word ``latest'' denotes maximality rather than a superlative.
For instance, among a finite set of events located on a spacelike 
hyperplane, each of them is a latest one.
A basic fact is that $n$ distinct events in Minkowski space always
have some latest member, and the same is true for sufficiently small
neighborhoods of them. Again, this is not so for strings, see the
counterexample in Fig.~\ref{fg:big-three} once more.

Lemma~\ref{WeakTrans} implies that two comparable strings and one
event, disjoint from the strings, always have a latest member.
(Namely, if $x$ is later than both $S_1$ and $S_2$ then it is of
course $x$, and otherwise it is the later of the two strings.) For the
desired chopping of $n > 2$ strings, we need a bit more.

\begin{lema} 
\label{StringsPoints}
Let $S_1,\dots,S_r$ be strings among which there is a latest one, and
let $y_1,\dots,y_k$ be pairwise distinct events in Minkowski space
satisfying $y_i \notin S_j$ for all $i = 1,\dots,k$ and 
$j = 1,\dots,r$. Then, the set 
$\{S_1,\dots,S_r,\{y_1\},\dots,\{y_k\}\}$ also has a latest member.
\end{lema}

\begin{proof}
We assume that $S_1$ is a latest member of the given strings.

\emph{First case}:
One of the points, say $y_1$, is later than all the strings. Let $y_l$
be a latest member of $\overline{V_+(y_1)} \cap \{y_1,\dots,y_r\}$,
the subset of $y$'s which lie to the future of~$y_1$. Then $y_l$ is a
latest member of \emph{all} the $y$'s, and also later than all the
strings,%
\footnote{Here we use the obvious fact that $\big(y_1 \later S
\wedge y_2 \in \overline{V_+(y_1)}\big) \Rightarrow y_2 \later S$.}
that is, $y_l$ is a latest member of
$\{S_1,\dots,S_r,\{y_1\},\dots,\{y_k\}\}$.

\emph{Second case}: 
No $y_i$ is later than all the strings. Then for every 
$i \in \{1,\dots,k\}$ there is a $j(i) \in \{1,\dots,r\}$ such that
$y_i \not\later S_{j(i)}$. If $j(i) = 1$, the label of the latest
string, then $y_i \not\later S_1$,
and Lemma~\ref{comparable} implies that $S_1 \later y_i$. If 
$j(i)\neq 1$, then $S_1 \later S_{j(i)}$ and Lemma~\ref{WeakTrans}
implies that $S_1 \later y_i$. Summarizing, in the second case for all
$i$ there holds $S_1 \later y_i$. Then $S_1$ is a latest member of
$\{S_1,\dots,S_r,\{y_1\},\dots,\{y_k\}\}$.
\end{proof}

\subsection{String chopping} 
\label{ssc:chop-shop}

As mentioned in the introduction, we wish to show that one can chop
$n$ strings into small enough pieces which are mutually comparable. 
We first give a constructive prove for $n = 2$, where it suffices to  
cut one of the strings once.

By \textit{cutting} of a string $S = \String{x,e}$ is meant the
selection of one point $x + se$ for some $s > 0$, whereby the string
becomes a union of the finite segment $x + [0,s]\,e$ and the residual
string $x + [s,\infty)\,e$. The two pieces do not overlap, since they
have only the cut point in common. We write this nonoverlapping union
as $S = S^1 \cup S^2$, but it suits us not to specify which piece is
the finite segment and which is the tail.

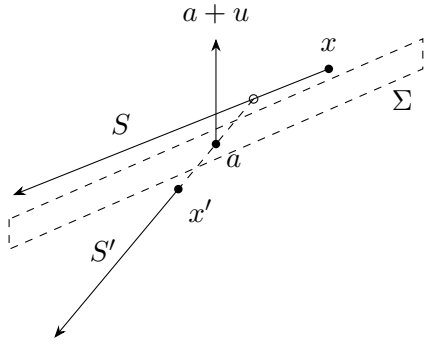
\begin{figure}[htb]
\centering
\begin{tikzpicture}[>=Stealth]
\coordinate (O) at (0,0) ; 
\coordinate (A) at (-0.5,-0.6) ;
\coordinate (U) at (-0.5,0.8) ;  
\coordinate (X) at (1,0.4) ;
\coordinate (Sa) at (-3.25,-1.6) ; \coordinate (Sb) at (-3.25,-2) ;
\coordinate (Sc) at (2.25,0.4) ; \coordinate (Sd) at (2.25,0.8) ;
\coordinate (Xa) at ($ (X)!2.5!(O) $) ;
\coordinate (Xb) at ($ (X)!4.2!(O) $) ;
\coordinate (Y) at ($ (O)!2.0!(A) $) ;  
\coordinate (Ya) at ($ (O)!4.0!(A) $) ;
\coordinate (Yb) at ($ (O)!5.3!(A) $) ;
\draw[dashed] (O) -- (Y) ;
\draw[->] (A) node[below right] {$a$} -- (U) node[above=2pt] {$a+u$} ;
\draw[->] (X) node[above=3pt] {$x$} 
   -- (Xa) node[above left] {$S$} -- (Xb) ; 
\draw[->] (Y)  node[below right] {$x'$} 
   -- (Ya) node[above=3pt] {$S'$} -- (Yb) ;
\draw[dashed] (Sa) -- (Sb) -- (Sc) node[below=12pt, left] {$\Sigma$} 
   -- (Sd) -- cycle ;
\foreach \pt in {A,X,Y} \fill (\pt) circle(1.6pt) ;
\draw (O) circle(1.6pt) ;
\end{tikzpicture} 
\caption{The strings $S$ and $S'$, with $\widetilde{S'}$ meeting $S$}
\label{fg:rear-window} 
\end{figure}

\begin{prop} 
\label{StringChopping2}
Let $S,S'$ be two disjoint strings. Then there is a cutting of $S$
into two pieces $S = S^1 \cup S^2$, such that both pairs $(S^1,S')$
and $(S^2,S')$ are comparable.
\end{prop}

\begin{proof}
Let $S = \String{x,e}$ and $S' = \String{x',e'}$, and denote by
$\widetilde{S'} \doteq \String{x',e'} \cup \String{x',-e'}$ the full
straight line through $x'$ with direction~$e'$.

We first consider the case when $S$ meets $\widetilde{S'}$ (but is
disjoint from~$S'$). Then there are positive reals $t,t'$ such that
\begin{equation} 
\label{eqtt'} 
x + te = x' - t'e'.
\end{equation}
Suppose the linear span of $e,e'$ is spacelike. Then $S$ and $S'$ are
disjoint sets contained in the spacelike $2$-plane 
$x + \linspan\{e,e'\}$. This implies that $S$, $S'$ are spacelike
separated, and thus comparable, see Lemma~\ref{R>S}(i). (No chopping
is needed.)

Suppose now that the span of $e,e'$ is timelike or lightlike:
see Appendix~\ref{app:BasicGeo}. Then
Lemma~\ref{Spanee'} implies that one of the vectors $e \pm e'$ is
timelike or lightlike. Suppose first that $e - e'$ is a timelike or
lightlike vector, and assume that it is future oriented,
$e \in \overline{V_+(e')}$. Let $u$ be any timelike future-oriented
vector orthogonal to~$e$, and (see Fig.~\ref{fg:rear-window}) let
$$
\Sigma \doteq a + u^\perp,  \quad  a \doteq x' - \half t' e'.
$$
Now $e' \cdot u$ is strictly negative since $e \cdot u = 0$ and
$e' \in \overline{V_-(e)}$, and we therefore get  
\begin{alignat}{2}
(x + se - a) \cdot u &= - \half t' e' \cdot u  && > 0,
\notag \\
(x' + s'e' - a) \cdot u &= (s' + \half t')e' \cdot u && < 0,
\label{eq:better-or-worse} 
\end{alignat}
for all $s,s' \geq 0$.
In the first line we have used Eq.~\eqref{eqtt'}.
These two inequalities say that $S \later \Sigma$ and that 
$S' \earlier \Sigma$, respectively. By Lemma~\ref{SSHyperplane}, this
shows that $S \later S'$. If $e - e'$ is past oriented,
$e \in \overline{V_-(e')}$, then the same argument shows that
$S' \later S$.
In the case when $e + e'$ (rather than $e - e'$) 
is a timelike or lightlike vector, the same argument goes through
where now $e' \cdot u > 0$ on the right hand side 
of~\eqref{eq:better-or-worse}.
This completes the proof in the case when $S$ meets~$\widetilde{S'}$. 

For the rest of the proof we consider the case when $S$ is disjoint
from~$\widetilde{S'}$. 
First, assume that $S \cap (\widetilde{S'})^c = \emptyset$, that is,
$S$ is contained in the closure of 
$V_-(\widetilde{S'}) \cup V_+(\widetilde{S'})$. If $S$ had nontrivial
intersection with both $\overline{V_-(\widetilde{S'})}$ and
$\overline{V_+(\widetilde{S'})}$, it would have to pass through
$\widetilde{S'}$, which was excluded. Thus, $S$ is contained entirely
in the closure of either $V_-(\widetilde{S'})$ or
$V_+(\widetilde{S'})$, and in this case $S$ and $S'$ are comparable by
Lemma~\ref{R>S}(ii). No chopping is needed.

Now suppose that $S \cap (\widetilde{S'})^c \neq \emptyset$. If
$e = \pm e'$ (i.e., the strings $\widetilde S$ and $\widetilde{S'}$
are parallel), then $S$ is completely contained in the causal
complement of $\widetilde{S'}$, and thus $S$ is both later and earlier
than $S'$, and no chopping is needed.

Consider finally the case $e \neq \pm e'$. The claim is that there
exists a cutting $S = S_+ \cup S_-$, such that 
$S_+ \later S' \later S_-$.
Using Lemma~\ref{SSHyperplane}, it is sufficient to establish the
existence of two spacelike hyperplanes $\Sigma_1$, $\Sigma_2$, such
that $S_+ \later \Sigma_1 \later S'$ and 
$S' \later \Sigma_2 \later S_-$. 

\begin{figure}[htb]
\centering
\begin{tikzpicture}[>=Stealth]
\coordinate (O) at (0,0,0) ; 
\coordinate (A) at (-1.4,0.5,0) ;
\coordinate (B) at (2.6,0,2.6) ;
\coordinate (Cl) at (-1.5,2,0) ; \coordinate (Cr) at (1.5,2,0) ;
\coordinate (Clm) at (-1.5,-2,0) ; \coordinate (Crm) at (1.5,-2,0) ;
\coordinate (P) at (0,1.8,0) ; \coordinate (Pm) at (0,-1.8,0) ;
\coordinate (Sa) at (-3.2,-0.7,-3.2) ; 
\coordinate (Sb) at (-3.2,0.7,-3.2) ;
\coordinate (Sc) at (3.6,0.7,3.6) ; 
\coordinate (Sd) at (3.6,-0.7,3.6) ;
\coordinate (U) at (0,1,0) ;
\coordinate (X) at (-0.4,1,0) ; 
\coordinate (Xa) at ($ (X)!2.0!(A) $) ;
\coordinate (Z) at ($ (X)!2.8!(A) $) ;
\coordinate (Zp) at ($ (O)!1.2!(B) $) ;
\draw (Cl) -- (Crm)   (Clm) -- (Cr) ;
\draw[->] (O) -- (B) node[above=2pt] {$S'$} -- (Zp) ;
\draw[->] (O) node[left=3pt] {$x'$} -- (U) node[right=2pt] {$u+x'$} ;
\draw[->] (X) node[above right] {$x$} -- (A) node[above left] {$a$}
   -- (Xa) node[above left] {$S$} -- (Z) ; 
\draw (P) ellipse(1.3cm and 0.4cm) ;
\draw (Pm) ellipse(1.3cm and 0.4cm) ;
\draw[dashed] (Sa) -- (Sb) -- (Sc)
   -- (Sd) node[above right] {$\Sigma$} -- cycle ;
\foreach \pt in {A,O,X} \fill (\pt) circle (1.6pt) ;
\end{tikzpicture} 
\caption{Selection of a cut point on the string $S$}
\label{fg:side-window} 
\end{figure}
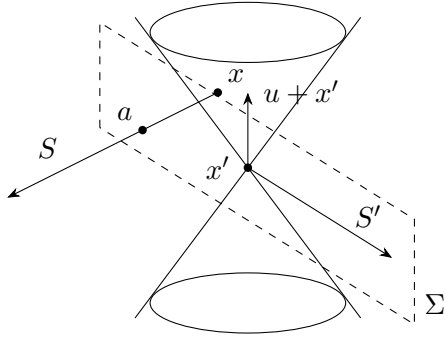

Take an event $a \in S \cap (\widetilde{S'})^c$ with $a \neq x$;
this is the place where we cut~$S$ (see Fig.~\ref{fg:side-window}). 
The vector $a - x'$ is spacelike and spacelike
separated from~$e'$, hence the $2$-plane 
$E \doteq \linspan\{a - x',e'\}$ is spacelike. Choose a timelike
future-directed vector $u$ in the orthogonal complement $E^\perp$,
which is \emph{not} orthogonal to $e$. (The possibility
$u \cdot e \neq 0$ is allowed since $e \neq \pm e'$.)
Note that $\Sigma \doteq x' + u^\perp$ contains the string $S'$ and 
meets the string $S$ at~$a$. Our hyperplanes $\Sigma_1$,
$\Sigma_2$ will be small modifications of~$\Sigma$. First, we shift
$\Sigma$ by a small amount so that it does not contain the point $a$
any more. Let $P_{e'}^\perp$ be the projector onto~$(e')^\perp$, and
let
$$
u_\pm \doteq u \pm \eps P_{e'}^\perp(a - x'),
$$
where $\eps$ is small enough that the sign of $u \cdot e$ is 
unchanged:
\begin{equation}
\sgn(u_\pm \cdot e) = \sgn(u \cdot e) \doteq \sigma.
\label{eq:good-signs} 
\end{equation}
Let now $\Sigma_\pm \doteq x' + (u_\pm)^\perp$, and define the
cutting $S = S_+ \cup S_-$, where
$$
S_\pm \doteq (a \pm \sigma\RR^+ e) \cap S.  
$$
(If $\sigma > 0$, then $S_+$ is the infinite tail of the cutting
while $S_-$ is a finite segment; if $\sigma < 0$, the roles are
interchanged.) Both $\Sigma_\pm$ still contain $S'$, and are 
moreover comparable with~$S_\pm$, as we now verify.

\begin{claim}
$S_+ \later \Sigma_-$ and $S_- \earlier \Sigma_+$\,.
\end{claim}

\begin{proof}[Proof of claim] 
Let $\xi \doteq a - x'$. By Lemma~\ref{Hyperplane}, the first relation is
equivalent to  
\begin{equation}
\label{eqS+Sigma-} 
(\xi + s\sigma e) \cdot u_- 
\equiv - \eps \xi \cdot P_{e'}^\perp(\xi) + s\,|e \cdot u_-| > 0 
\word{for all} s \geq 0.
\end{equation}
We have used $(a - x') \cdot u = 0$ and Eq~\eqref{eq:good-signs}. Note
that the projector $P_{e'}^\perp$ is given by
$$
P_{e'}^\perp(\xi) = \xi - \frac{e'\cdot\xi }{e'\cdot e'}\,e' 
= \xi + (e' \cdot \xi) e', 
$$
hence $\xi \cdot P_{e'}^\perp(\xi) = \xi \cdot \xi + (e'\cdot \xi)^2$.
Now the condition that $a \in (\widetilde{S'})^c$ means that
$\xi - te'$ is spacelike for all $t \in \RR$. Thus, the quadratic form
$- t^2 - 2t(e' \cdot \xi) + \xi \cdot \xi 
\equiv -(t + e' \cdot \xi)^2 + \xi \cdot \xi + (e' \cdot \xi )^2$ is
strictly negative, which implies that
\begin{equation}
(\xi \cdot \xi) + (e' \cdot \xi)^2 < 0;  
\word{and thus}  \xi \cdot P_{e'}^\perp(\xi ) < 0.
\label{eq:neg-term} 
\end{equation}
This proves the inequality~\eqref{eqS+Sigma-}, and thus the first
relation of the claim. The second one is shown analogously: it means
that 
\begin{equation}
\label{eqS-Sigma+} 
(\xi - s\sigma e) \cdot u_+ 
\equiv \eps \xi \cdot P_{e'}^\perp(\xi) - s\,|e \cdot u_-| < 0 
\word{for all} s \geq 0, 
\end{equation}
which holds true by Eq.~\eqref{eq:neg-term}. 
\end{proof}

We now shift the hyperplanes $\Sigma_\pm$ a little bit, so that they
still satisfy the relations of the above lemma, and are in addition
comparable with~$S'$.  
To this end, notice that the left hand side of the
inequality~\eqref{eqS+Sigma-} has the positive lower bound 
$\delta \doteq -\eps  \xi \cdot P_{e'}^\perp \xi $. Thus, we can shift 
$\Sigma_-$ away from $S'$ by using instead
$$ 
\Sigma_1 \doteq \Sigma_- + \alpha\,u_-  \word{where} 
\alpha \doteq \frac{\delta}{2u_- \cdot u_-}\, , 
$$ 
and the relation $S_+ \later \Sigma_1$ still holds. On the other hand,
for all $s > 0$ there holds
$$
(se' - \alpha u_-) \cdot u_- \equiv - \half\delta < 0, 
$$
and therefore $S' \earlier \Sigma_1$. We have now achieved 
$S_+ \later \Sigma_1 \later S'$, as required.

Similarly, one verifies that 
$\Sigma_2 \doteq \Sigma_+ - \frac{\delta}{2u_+ \cdot u_+}\,u_+$
satisfies $S_+ \earlier \Sigma_2 \earlier S'$. 
This completes the proof of Prop.~\ref{StringChopping2}.
\end{proof}

We now consider the case of $n > 2$ strings. The \emph{large string
diagonal} is defined by
\begin{equation} 
\label{eqLargeStringDiag} 
\LargeDiag \doteq \set{ (x_1,e_1,\dots,x_n,e_n)
: \String{x_i,e_i} \cap \String{x_j,e_j} \neq \emptyset 
\text{ for some } i \neq j}. 
\end{equation}
We are going to show that $n$ strings outside $\LargeDiag$ can be
chopped up into finitely many pieces which are mutually comparable
(Prop.~\ref{StringChopping}). Here we shall need to cut the strings
into more than two pieces.
By a \emph{chopping} of a string $S \doteq x + \RR_0^+ e$ we mean a
decomposition
\begin{equation} 
\label{eqChop} 
S = \Sfin \cup \Sinfty 
= S^1 \cup S^2 \cup\cdots\cup S^N \cup S^\infty,
\end{equation}
determined by $N$ numbers $0 = s_0 < s_1 < \cdots < s_N$, 
where $S^1,\dots,S^N$ are \emph{consecutive} nonoverlapping finite
segments, 
$S^\alpha \doteq  x + [s_{\alpha-1}, s_\alpha]\,e$
and $\Sinfty \doteq x + [s_N,\infty)\,e$ is the infinite tail of the
original string. We find it convenient to write
$S^{N+1} \doteq S^\infty$, too.

Before stating and proving Prop.~\ref{StringChopping} below, we need a
few lemmas. We start by considering the infinite tails of the
strings $\String{x_i,e_i}$. If one looks at $n$~strings from
sufficiently far away, their ``heads'' $x_i$ appear quite close to the
origin (wherever the origin may be located). Hence, on cutting them
far away from their heads, their infinite tails extend almost
radially to infinity and thus correspond to points on the hyperboloid
$\Spd$ of spacelike directions. Consequently, those tails can be
linearly ordered, much like events in~$\Spd$.

We realize this idea by showing first that every string $\String{x,e}$
eventually ends up in a spacelike cone centered around the string
$\String{0,e}$ with arbitrarily small opening angle. In detail, let
$D$ be a neighborhood of $e$ in $\Spd$, and let $C_D$ be the spacelike
cone centered at the origin:
\begin{equation} 
\label{eqSpc} 
C_D \doteq \RR^+ D = \set{se' : s > 0, \ e' \in D}.
\end{equation}

\begin{lema} 
\label{InfiniteTail}
For every string $\String{x,e}$ and every neighborhood $D$ of $e$
in~$\Spd$, there is an $s > 0$ such the infinite tail
$x + [s,\infty)\,e$ is contained in the spacelike cone~$C_D$.
\end{lema}

\begin{proof}
Note that $y \in C_D$ if and only if 
$|y \cdot y|^{-1/2}y$ lies in $D$. Thus, a point $x + te$ on the
string is in $C_D$ if and only if the point 
$|(x + te) \cdot (x + te)|^{-1/2}(x + te)$ is in the
neighborhood~$D$. But this point can be written as
$$
\frac{t}{|x \cdot x + 2t\,x\cdot e - t^2|^{1/2}}\,
\Bigl( \frac{x}{t} + e \Bigr),  
$$
which obviously converges to $e$ as $t \to \infty$.
Thus, the curve $|(x + te)^2|^{-1/2}(x + te)$ in~$\Spd$
eventually ends up in~$D$, that is, $x + te$ eventually ends up
in~$C_D$, as claimed.
\end{proof}

The previous Lemma is relevant for time ordering due to the following
fact.

\begin{lema} 
\label{InfiniteTailSpd}
Take two strings $S_1$, $S_2$ which are contained in spacelike cones
of the form \eqref{eqSpc}, $S_i \subset C_{D_i}$, where $D_1$ and
$D_2$ are double cones in the manifold $\Spd$ of spacelike directions.
Suppose $D_1 \later D_2$, where the posteriority ordering on~$\Spd$ is
defined in the same way as that of Minkowski space, see
Eq.~\eqref{eqTOGeo}. Then $S_1 \later S_2$.
\end{lema} 

\begin{proof}
Just as in Minkowski space (see Def.~\ref{DoCo}), each double cone
$D_i$ is characterized by its past and future tips, 
$e_i^+ \in V_+(e_i^-)$:  
$$
D_i = D_i(e_i^-, e_i^+) \cap \Spd
\doteq  V_+(e_i^-) \cap V_-(e_i^+) \cap \Spd. 
$$
The hypothesis $D_1 \later D_2$ obviously implies (in fact, is
equivalent to) $e_1^- \later e_2^+$. To proceed, we first need an
intermediate result: not only that there exists a spacelike hyperplane
$\Sigma$ such that $e_1^- \later \Sigma \later e_2^+$, but that there
is even one passing through the origin that does so.

\begin{claim}
Let $e,e' \in \Spd$ with $e \later e'$. Then there exists $u \in V_+$
such that $u \cdot e > 0$ and $u \cdot e' < 0$.
\end{claim}

\begin{proof}[Proof of claim] 
The following four cases can occur:  
\begin{enumerate}
\item 
The span of $e,e'$ is spacelike. Then, by Lemma~\ref{Spanee'},
$e \cdot e' \in (-1,1)$. Let $u \in $ span$\{e,e'\}^\perp$ be a
future-pointing timelike vector and define 
$u_\eps \doteq u - \eps(e - e')$ with $\eps > 0$ small enough so that
$u_\eps$ is still in~$V_+$. Then
$$
u_\eps \cdot e = \eps(1 + e \cdot e') > 0  \word{and} 
u_\eps \cdot e' = - \eps(1 + e \cdot e') < 0.
$$

\item 
The span of $e,e'$ is timelike. According to Lemma~\ref{Spanee'}, the
following two subcases can occur:
\begin{enumerate}
\item 
The vector $e - e'$ is timelike. 
Then, since $e \later e'$ is assumed, it is future-pointing. Moreover,
we must have $e \cdot e' < -1$. Thus, the vector $u \doteq e - e'$
does the job: $u \cdot e = -1 - e \cdot e' > 0$ and 
$u \cdot e'= e \cdot e' + 1 < 0$.

\item 
The vector $e - e'$ is spacelike and $e + e'$ is timelike. Then we
must have $e \cdot e' > 1$. If $e + e'$ is future pointing, then
$$
u \doteq P_{e'}^\perp(e) = e + (e' \cdot e)e' 
$$
is timelike, since $u \cdot u = -1 + (e \cdot e')^2 > 0$. It is also
future-pointing. (That can be seen as follows: choose $v \in V_+$ with
$v \cdot e = 0$; then $v \cdot e' \equiv v \cdot (e + e') > 0$ since
$e + e'$ is future-pointing, and so also $v \cdot u > 0$.)
Now put $u_\eps \doteq u + \eps e'$ for sufficiently small $\eps > 0$.
Then we have
$$
u_\eps \cdot e = -1 + (e' \cdot e)^2 + \eps e \cdot e' > 0
\word{and}  u_\eps \cdot e' = -\eps < 0.  
$$
If $e + e'$ is past-pointing, then 
$$
u_\eps \doteq - P_{e}^\perp(e') - \eps e 
= -\bigl( e' + (e \cdot e')e \bigr) - \eps e 
$$
has the following properties, as the reader may readily verify: it is
timelike and future-pointing, and satisfies
$$
u_\eps \cdot e = \eps > 0  \word{and} 
u_\eps \cdot e' = 1 - (e' \cdot e)^2 - \eps e \cdot e' < 0.
$$
\end{enumerate} 

\item 
The span of $e,e'$ is lightlike. From Lemma~\ref{Spanee'}, there are
two possibilities:
\begin{enumerate}
\item 
The vector $e - e'$ is lightlike. Then it is future pointing by
hypothesis, and moreover we must have $e \cdot e' = -1$. Choose
$u \in (e')^\perp \cap V_+$ and let $u_\eps \doteq u + \eps e$ with
sufficiently small~$\eps$. Then $u_\eps$ is a future-pointing timelike
vector satisfying
$$
u_\eps \cdot e = u \cdot e - \eps > 0  \word{and} 
u_\eps \cdot e' = - \eps < 0.
$$
(We used that $u \cdot e \equiv u \cdot (e - e')$ is positive since 
$e - e'$ is assumed future-pointing; and we chose $\eps$ small
enough.)

\item 
The vector $e + e'$ is lightlike. Then we must have $e \cdot e' = 1$.
If the vector $e + e'$ future-pointing, the same $u_\eps$ as in the
previous item does the job. If it is past pointing, pick 
$u \in e^\perp \cap V_+$ and let $u_\eps \doteq u-\eps e$ with
sufficiently small $\eps$. Then $u_\eps$ is a future-pointing timelike
vector satisfying
$$
u_\eps \cdot e = \eps > 0  \word{and} 
u_\eps \cdot e' = u \cdot e' + \eps < 0.
$$
(We used that $u \cdot e' \equiv u \cdot (e + e')$ is negative since
$e + e'$ is assumed past-pointing; and we chose $\eps$ small enough.)
\end{enumerate} 

\item 
In the last possible case, $e = -e'$, let $u \in e^\perp \cap V_+$ and
define $u_\eps \doteq u - \eps e$. Then
$$
u_\eps \cdot e = \eps > 0  \word{and}  u_\eps \cdot e' = - \eps < 0.
$$
\end{enumerate}
This proves the claim. 
\end{proof}

We have thus shown that there exists a future-pointing timelike vector
$u$ that satisfies $u \cdot e_1^- > 0 > u \cdot e_2^+$. It follows
that for all $e_1 \in D_1$ and $e_2 \in D_2$, we have
$u \cdot e_1 > 0 > u \cdot e_2$. This implies of course that 
$u \cdot re_1 > 0 > u \cdot se_2$ for $r,s \in \RR^+$, and since all
$z_i \in C_{D_i}$ are of the form $z_i = re_i$ with $r \in \RR^+$ and
$e_i \in D_i$, we have $C_{D_1} \later u^\perp \later C_{D_2}$ and
consequently, by Lemma~\ref{SSHyperplane}, $S_1 \later S_2$.
\end{proof}

We are now prepared for our main geometrical result.

\begin{prop} 
\label{StringChopping}
Let $(\underline{x},\underline{e})$ be outside the large string
diagonal $\LargeDiag$. Then there exists a chopping
$$     
\String{x_i,e_i} = S_i^1 \cup\cdots\cup S_i^{N_i} \cup S_i^{N_i+1}
$$
such that every selection $\{S_1^{\alpha_1},\dots,S_n^{\alpha_n}\}$
has a latest member, that is, for every \mbox{$n$-tuple}
$(\alpha_1,\dots,\alpha_n)$ there exists
$i \in \{1,\dots,n\}$ 
such that for every $j \in \{1,\dots,n\} \setminus \{i\}$ the
relation $S_i^{\alpha_i} \later S_j^{\alpha_j}$ holds. 
\end{prop} 

\begin{proof}
We first consider the infinite tails. Note that some of the~$e_i$ may
coincide, so the \emph{set} of $e$'s may contain fewer than $n$
distinct points. These have a latest member, and the same holds for
sufficiently small double cones $D_i \ni e_i$ (understanding that
$D_i = D_j$ if $e_i = e_j$). Let us denote the index of the latest
member of $\{D_1,D_2,\dots\}$ by $i_0$. Let further $s_i \in \RR_0^+$
be such that the infinite tail
$S_i^\infty \doteq x_i + [s_i,\infty)\,e_i$ of $S_i$ is contained in
$C_{D_i}$, see Lemma~\ref{InfiniteTail}. By
Lemma~\ref{InfiniteTailSpd}, the infinite tail with number $i_0$ is
later than all the others. If $e_{i_0}$ coincides with
$\{e_{i_1},\dots,e_{i_k}\}$, then the corresponding infinite tails are
parallel and disjoint, and therefore have a latest member. 
(Since the problem can be reduced to distinct points in
$e_{i_0}^\perp$.) This is a latest member among all infinite tails.

We now construct a chopping of the compact segments 
$\Sfin_i \doteq x_i + [0,s_i]\,e_i$. Consider an arbitrary point
$(y_1,\dots,y_n)$ on $\Sfin_1 \xyx \Sfin_n$.
The events $y_i$ are all distinct and hence the set
$\{y_1,\dots,y_n\}$ has a latest member, and the same holds for
sufficiently small neighborhoods $U_0(y_i)$ of the $y_i$. Similarly,
for any subset $I \subset \{1,\dots,n\}$ with complement
$J \doteq \{1,\dots,n\} \setminus I$, the infinite tails $S_j^\infty$
and the events~$y_i$, $(i,j) \in I \x J$,
fulfil the hypotheses of Lemma~\ref{StringsPoints}, which state that
the $n$~sets $S_j^\infty$ and $\{y_i\}$ together have a latest member.
The same
holds for sufficiently small neighborhoods $U_I(y_i)$ of the points
$y_i$, $i \in I$. Let now $U(y_i)$ be the intersection of $U_0(y_i)$
and of all $U_I(y_i)$, where $I$ runs through the subsets of
$\{1,\dots,n\} \setminus \{i\}$. Then for any 
$I \subset \{1,\dots,n\}$, the $n$~sets $S_j^\infty$ and $U(y_i)$
together, for $(i,j) \in I \x J$, have a latest member. Of course the
same holds for the intersections of these neighborhoods with the
corresponding strings,
\begin{equation} 
\label{eqSegment}
V(y_i) \doteq U(y_i) \cap \Sfin_i .
\end{equation}
Summarizing, for each $(y_1,\dots,y_n) \in \Sfin_1 \xyx \Sfin_n$ there
is a neighborhood on the strings
$V(y_1) \xyx V(y_n) \subset \Sfin_1 \xyx \Sfin_1$
wherein for any $I \subset \{1,\dots,n\}$ the $n$~sets $S_i^\infty$
and $V(y_j)$, with $(i,j) \in I \x J$, have a latest member. Now for
each~$i$ the union
$$
\bigcup \set{V(y_i) : y_i \in \Sfin_i}
$$ 
is an open covering of the set $\Sfin_i$. 
By compactness, it has a finite subcovering. That is to say, on
the string segment $\Sfin_i$ there are finitely many events
$y_i^1,\dots,y_i^{N_i}$ such that the finite union
$$ 
V(y_i^1) \cup\cdots\cup V(y_i^{N_i})
$$ 
still covers $\Sfin_i$. Here we may assume that 
$y_i^1,\dots,y_i^{N_i}$ are consecutive events along the segment
$\Sfin_i$; and that $V(y_i^\alpha)$ and $V(y_i^\beta)$ overlap if and 
only if $\beta = \alpha \pm 1$. All
these neighborhoods still have a
latest member in the sense mentioned after Eq.~\eqref{eqSegment}.
Now choose, for each $\alpha \in \{1,\dots, N_i-1\}$, a number
$s_i^\alpha$ such that
$x_i + s_i^\alpha e_i \in V(y_i^\alpha) \cap V(y_i^{\alpha+1})$, 
and define
$$
S_i^\alpha \doteq x_i + [s_i^\alpha, s_i^{\alpha+1}]\, e_i. 
$$
Then $S_i^\alpha$ is included in $V(y_i^\alpha)$, and hence each
$n$-tuple of string segments or infinite tails 
$S_1^{\alpha_1}, \dots, S_n^{\alpha_n}$ has a latest member, as
claimed. This concludes the proof of Prop.~\ref{StringChopping}.
\end{proof}

\section{Time-ordered products of linear fields} 
\label{sec:TimeOrderedProd}

We consider quantum fields of the form~\eqref{eqFieldSt}, not fixing
the ``weight function'' $u(s)$ and admitting the case when the weight
function has support in a proper interval $I \subset \RR_0^+$. In this
case the field is localized on the (possibly finite) string segment
$x + Ie$. The same holds after multiplication of the field with a
$C^\infty$ function $f(x,e)$. The fields of this form, with $u$
varying and multiplied by $C^\infty$ functions, generate a linear
space of operator valued distributions, which we denote by $\calL$ and
call the space of ``linear fields''.%
\footnote{For example, if we admit among the point-localized fields
also vector fields $A_\mu(x)$, then the space $\calL$ contains fields
of the form $\int_{s_1}^{s_2} ds\, A_\mu(x + se)\, e^\mu$.}

We now set out to define the time-ordered products
\begin{equation} 
\label{eqTn} 
\T_n \fieldSt_1(x_1,e_1) \cdots \fieldSt_n(x_n,e_n)
\end{equation}
of linear fields $\fieldSt_i \in\calL$.   
These are operator-valued distributions on $(\RR^4 \x \Spd)^{\x n}$
acting on the domain $\calD$ of vectors with finite particle number
and smooth momentum-space wave functions, which are required to share
the following properties.
\begin{itemize} 
\refstepcounter{P} 
\item[\theP] 
\label{T1} \refstepcounter{P} 
$\T_1$ is given by $\T_1\fieldSt(x,e) \doteq \fieldSt(x,e)$. 
\item[\theP] 
\label{TLinearity} \refstepcounter{P} \emph{Linearity}. 
The time-ordered product $\T_n$ is an $n$-linear mapping from the
space $\calL$ of linear fields into operator-valued distributions
acting on~$\calD$. 
\item[\theP] 
\label{TSymmetry} \refstepcounter{P} \emph{Symmetry}. 
$\T_n$ is totally symmetric in its $n$ arguments.
\item[\theP] 
\label{TCausality} \refstepcounter{P} \emph{Causality}.
Let $\fieldSt_i$ be localized on the string (or string segment) $S_i$,
$i = 1,\dots,n$. If $S_i \later S_j$ for all $i \in \{1,\dots,k\}$ and
$j \in \{k+1,\dots,n\}$, then the following factorization holds:
\begin{align*}
& \T_n \fieldSt_1(x_1,e_1) \cdots \fieldSt_n(x_n,e_n)
\\
&\qquad = \T_k \fieldSt_1 (x_1,e_1) \cdots \fieldSt_k(x_k,e_k)\,
\T_{n-k} \fieldSt_{k+1}(x_{k+1},e_{k+1}) \cdots \fieldSt_n(x_n,e_n).
\end{align*}
\end{itemize}
In the sequel, we shall generically denote fields in $\calL$ by
$\fieldSt(x,e)$, without further specifications. 

Before we turn to the construction of the $T$-products, we recall
Wick's theorem for linear fields, which is also valid in the
string-localized case [here $\field(i)$ denotes $\field(x_i)$ or
$\field(x_i,e_i)$ in the string-localized case]:
\begin{equation}
\label{eqWickLinear} 
\field(1) \cdots \field(n) 
= \sum_G \prod_{l\in\Eint} 
\Scalprod{\Omega}{\field(s(l)) \field(r(l))\,\Omega}\,
\Wickli \prod_{l\in\Eext} \field(s(l)) \Wickre.  
\end{equation}
Here, $\Omega$ denotes the vacuum vector and $\scalprod{\cdot}{\cdot}$
denotes the scalar product. $G$ runs through the set of all 
graphs with vertices $\{1,\dots,n\}$ and oriented lines, such that
from every vertex there emanates 
one line. The lines either connect two
vertices (internal lines, $l \in \Eint$) or go from a vertex to the
exterior (external lines, $l \in \Eext$). The initial vertex of an
internal line $l$ (source $s(l)$) has a smaller index than its final
vertex (target $r(l)$). The external lines only have sources.

Let us recall how the time-ordered products are constructed in the
point-local case. In a first step, one shows that Wick's
expansion~\eqref{eqWickLinear} also holds for the time-ordered
products outside the large diagonal $\{x_i \neq x_j\}$, namely:
\begin{equation}
\label{eqTWickLinearPt} 
\T \fieldSt(1) \cdots \fieldSt(n)
= \sum_G \prod_{l\in\Eint} 
\Scalprod{\Omega}{\T \field(s(l)) \field(r(l))\,\Omega}\, 
\Wickli \prod_{l\in\Eext} \field(s(l)) \Wickre. 
\end{equation}
(The vacuum expectation value 
$\Scalprod{\Omega}{\T \field(x) \field(y)\,\Omega}$ is called the
\emph{Feynman propagator}.) This is shown by induction, using 
that $n$ distinct points always have a latest member in the 
posteriority sense.
In a second step, one constructs the extension across the large
diagonal (requiring certain (re-)normalization conditions). If the
scaling degree of the Feynman propagator is less than~$4$, then the
$T$-products are fixed (on all $\RR^{4n}$), namely, they are given by
Eq.~\eqref{eqTWickLinearPt}.
On the other hand, if the scaling degree of the Feynman propagator is
$\geq 4$ one may add, depending on the scaling degree and the number
of internal lines of the graph in~\eqref{eqTWickLinearPt},
renormalization terms in the form of delta distributions (and their
derivatives) in the difference variables with ``internal'' indices.
This is the case for fields with spin $\geq 1$ acting on a Hilbert
space.

We show here that for string-localized fields $\field(x,e)$ the $T_n$
are fixed outside the large string diagonal $\LargeDiag$ just by the
geometrical time-ordering prescription, namely they are given by the
same expression~\eqref{eqTWickLinearPt} as in the point-like case. As
mentioned in the introduction, the problem we have to overcome is the
fact that the set of points in $(\RR^4 \x \Spd)^{\x n}$ which
correspond to strings that are \emph{not} comparable in the sense
of~$\later$ is much larger than $\LargeDiag$, in fact it contains an
open set. We use our results on string chopping from the last section
to show that they are nevertheless fixed outside~$\LargeDiag$.
Recall that we are dealing with string-localized fields that can be
written as line integrals over point-localized fields as in
Eq.~\eqref{eqFieldSt}. Thus, for any chopping of the string
$\String{x,e} = \bigcup_\alpha S^\alpha$ as in Eq.~\eqref{eqChop}, the
field $\fieldSt(x,e)$ can be written as a sum
\begin{equation} 
\label{eqFieldSum} 
\fieldSt(x,e) = \sum_{\alpha=1}^{N+1} \fieldSt^{\alpha}(x,e)\,,
\word{where} 
\fieldSt^\alpha(x,e) 
\doteq \int_{s_{\alpha-1}}^{s_{\alpha}} ds\, u(s)\, \fieldPt(x + se)
\end{equation}
is localized on the segment $S^\alpha$. (We put $s_0\doteq 0$ and
$s_{N+1} \doteq \infty$.)

We start with two fields. 
If the two strings $\String{x,e} \doteq S$ and 
$\String{x',e'} \doteq S'$ are comparable, then \ref{TCausality}
implies that
\begin{equation} 
\T \fieldSt(x,e) \fieldSt(x',e') = \begin{cases}
\fieldSt(x,e)\,\fieldSt(x',e') &\text{if } S \later S', \\
\fieldSt(x',e')\,\fieldSt(x,e) &\text{if } S \earlier S'.
\end{cases}
\label{eqT2} 
\end{equation}
This is well defined, for if $S$ is both later and earlier than $S'$
then it is spacelike separated from $S'$ by Lemma~\ref{R>S} and 
the fields commute, so that both lines in \eqref{eqT2} are valid. 
The problem is that there is an open set of pairs of strings which are
not comparable, namely whenever $S$ meets both the past and the future
of~$S'$. This is solved by the procedure of string chopping, which fixes
the $T$-product outside the string diagonal.

\begin{prop} 
\label{T2Diag} 
The time ordered product $\T \field(x,e) \field(x'e')$ is uniquely
fixed outside the string diagonal $\Delta_2$ by \ref{T1} through
\ref{TCausality}. It satisfies Wick's expansion
\begin{equation} 
\label{eqTWick2} 
\T \field(x,e) \field(x',e') 
= \Wickli \field(x,e) \field(x',e') \Wickre
+ \Scalprod{\Omega}{\T \field(x,e) \field(x',e')\,\Omega}. 
\end{equation}
\end{prop}

\begin{proof}
If the two strings $\String{x,e} \doteq S$ and 
$\String{x',e'} \doteq S'$ are comparable, their $T$-product has been
defined in Eq.~\eqref{eqT2}. If the strings are not comparable, then
we cut one string, say $S$, into two pieces $S = S^1 \cup S^2$ such
that the pairs $(S^1,S')$ and $(S^2,S')$ are comparable, see
Prop.~\ref{StringChopping2}. As explained in Eq.~\eqref{eqFieldSum},
the field $\fieldSt(x,e)$ can be written as a sum 
$\fieldSt = \fieldSt^1 + \fieldSt^2$, where the field 
$\fieldSt^\alpha$ is localized on $S^\alpha$, $\alpha = 1,2$. By 
linearity \ref{TLinearity} of the $T$-product, we have
\begin{equation} 
\label{eqS1S2} 
\T \fieldSt(x,e) \fieldSt(x',e') 
= \T \fieldSt^1(x,e) \fieldSt(x',e') 
+ \T \fieldSt^2(x,e) \fieldSt(x',e'), 
\end{equation}
where both terms are fixed as in Eq.~\eqref{eqT2}. 

We need to show independence of the chosen chopping. Given a different
chopping $S = \widetilde S^1\cup\widetilde S^2$, one of the new pieces
$\tilde S^\alpha$ is contained in one of the old pieces, $S^\beta$. We
may assume that $\widetilde S^1 \subset S^1$. Then we have
\begin{equation} 
\label{eqS12} 
S^1 = \widetilde S^1 \cup S^{12}, \quad 
\widetilde S^2 = S^{12} \cup S^2, 
\end{equation}
where $S^{12} \doteq S^1 \setminus \widetilde S^1$ is the ``middle
piece'':
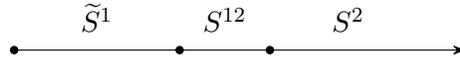
$$
\begin{tikzpicture}[>=Stealth]
\coordinate (O) at (0,0) ; 
\coordinate (A) at (2.2,0) ;
\coordinate (B) at (3.4,0) ;
\coordinate (Z) at (6,0) ;
\coordinate (S1) at ($ (O)!0.5!(A) $) ; 
\coordinate (S12) at ($ (A)!0.5!(B) $) ;
\coordinate (S2) at ($ (B)!0.4!(Z) $) ; 
\draw[->] (O) -- (S1) node[above=3pt] {$\widetilde{S}^1$} 
   -- (S12) node[above=3pt] {$S^{12}$} 
   -- (S2)  node[above=3pt] {$S^2$} -- (Z) ; 
\foreach \pt in {A,B,O} \fill (\pt) circle (1.6pt) ;
\end{tikzpicture} 
$$
The field decomposes as 
$\fieldSt  = \tilde{\fieldSt}^1 + \tilde{\fieldSt}^2$, where the
operator $\tilde{\fieldSt}^\alpha$ is localized on
$\widetilde S^\alpha$, $\alpha = 1,2$; and by Eq.~\eqref{eqS12} we
have
$$
\tilde{\fieldSt}^2(x,e) = \fieldSt^{12}(x,e) + \fieldSt^2(x,e)
\word{and} 
\tilde{\fieldSt}^1(x,e) + \fieldSt^{12}(x,e) = \fieldSt^1(x,e), 
$$
where $\fieldSt^{12}(x,e)$ is localized on the middle piece~$S^{12}$.
Notice that, by construction, $S^{12}$ is comparable with $S'$ since
it is contained in $S^1$ (or $\widetilde S^2$), and thus
$\T \fieldSt^{12}(x,e) \fieldSt(x',e')$ is fixed by Eq.~\eqref{eqS12}. 
With respect to the new chopping, therefore,
\begin{align*}
\T \fieldSt(x,e) \fieldSt(x',e') 
&= \T \tilde{\fieldSt}^1(x,e) \fieldSt(x',e') 
+ \T \tilde{\fieldSt}^2(x,e) \fieldSt(x',e')
\\
&= \T \tilde{\fieldSt}^1(x,e) \fieldSt(x',e')
+ \T \fieldSt^{12}(x,e) \fieldSt(x',e')
+ \T \fieldSt^2(x,e) \fieldSt(x',e')
\\
&= \T \fieldSt^1(x,e) \fieldSt(x',e')
+ \T \fieldSt^2(x,e) \fieldSt(x',e'). 
\end{align*}
This confirms independence of the chosen chopping in
Eq.~\eqref{eqS1S2}; we have shown uniqueness outside~$\LargeDiag$.
(If the roles of $S$ and~$S'$ are reversed, the
same conclusion applies; and by comparing either string cutting to a
chopping where both strings are cut, the linearity of the
$\T$-products maintains the uniqueness.)
On substituting Eq.~\eqref{eqT2} into Eq.~\eqref{eqS1S2} and applying
Wick's theorem for ordinary products, we get Wick's
expansion~\eqref{eqTWick2} for the $T$-products.
\end{proof}

We turn to the case of $n > 2$ fields, 
and show that Wick's expansion~\eqref{eqTWickLinearPt} also holds for
string-localized fields -- outside the large string diagonal.

\begin{prop} 
\label{TWickLinearDiag}
The time-ordered $n$-fold product of a string-localized free field
$\fieldSt(x_i,e_i)$ is uniquely fixed outside the large string
diagonal $\LargeDiag$, namely there holds
\begin{align} 
\label{eqTWickLinear} 
& \T \fieldSt(x_1,e_1) \cdots \fieldSt(x_n,e_n)
\nonumber \\
&\qquad = \sum_G \prod_{l\in\Eint} 
\Scalprod{\Omega}{\T \fieldSt(x_{s(l)},e_{s(l)})
\fieldSt(x_{r(l)},e_{r(l)})\,\Omega} \,
\Wickli \prod_{l\in\Eext} \fieldSt(x_{s(l)},e_{s(l)})\Wickre 
\end{align}
outside the large string diagonal. (Same notation as above.) 
\end{prop}

\begin{proof}
Let $(x_0,e_0,\dots,x_n,e_n)$ be outside the large string diagonal. 
That means that 
the strings $S_i \doteq \String{x_i,e_i}$ are mutually disjoint, for
$i = 0,\dots,n$. We wish to determine 
$T_{n+1} \doteq \T \fieldSt(0) \cdots \fieldSt(n)$, where we have
written $\fieldSt(i) \doteq \fieldSt(x_i,e_i)$, under the induction
hypothesis that the formula~\eqref{eqTWickLinear} is valid for
$T_n = \T \fieldSt(1) \cdots \fieldSt(n)$. Choose a chopping of the
$n + 1$ strings as in Prop.~\ref{StringChopping}, and let
$\fieldSt(i) = \sum_{\alpha=1}^{N_i+1} \fieldSt^{\alpha}(i)$ be the
corresponding decomposition as in Eq.~\eqref{eqFieldSum}. Then by
linearity~\ref{TLinearity},
$$
T_{n+1} = \sum_{\alpha_0,\dots,\alpha_n}
\T \fieldSt^{\alpha_0}(0) \cdots \fieldSt^{\alpha_n}(n) .
$$
For given $(\alpha_0,\dots,\alpha_n)$, denote by $i_0$ the index of
the latest member of the set of string segments
$\{S_0^{\alpha_0},\dots,S_n^{\alpha_n}\}$ as in
Prop.~\ref{StringChopping}. Then by \ref{TSymmetry}
and~\ref{TCausality},
$$
T_{n+1} 
= \sum_{\alpha_0,\dots,\alpha_n} \fieldSt^{\alpha_{i_0}}(i_0)\, 
\T \prod_{i\in I} \fieldSt^{\alpha_i}(i), 
$$
where we have written $I \doteq \{0,\dots,n\} \setminus \{i_0\}$. 
By the induction hypothesis, this is 
$$
\sum_{\alpha_0,\dots,\alpha_n} \fieldSt^{\alpha_{i_0}}(i_0) \, 
\sum_G \prod_{l\in\Eint} 
\bigl< \T \fieldSt^{\alpha_{s(l)}}(s(l))
\fieldSt^{\alpha_{r(l)}}({r(l)}) \bigr> \, 
\Wickli \prod_{l\in\Eext} \fieldSt^{\alpha_{s(l)}}({s(l)}) \Wickre, 
$$
where $G$ runs through all graphs $\calG(I)$ with vertices~$I$, and
$\langle\cdot\rangle$ denotes the vacuum expectation value. Using
Wick's Theorem for ordinary products, we have
\begin{multline*}
\fieldSt^{\alpha_{i_0}}(i_0) \,
\Wickli \prod_{i\in\Vext} \fieldSt^{\alpha_i}(i) \Wickre 
\\
= \Wickli \fieldSt^{\alpha_{i_0}}(i_0) 
\prod_{i\in\Vext} \fieldSt^{\alpha_i}(i) \Wickre
+ \sum_{i\in \Vext} 
\bigl< \fieldSt^{\alpha_{i_0}}(i_0) \fieldSt^{\alpha_i}(i) \bigr> \,
\Wickli \fieldSt^{\alpha_{i_0}}(i_0)
\prod_{j\in\Vext\setminus\{i\}} \fieldSt^{\alpha_j}(j) \Wickre \,, 
\end{multline*} 
where $\Vext$ denotes the set of vertices with external lines,
$\set{s(l) : l \in \Eext}$. 

Now since $i_0$ is the latest member of the string segments, we may 
replace the vacuum expectation value by the time-ordered one,
$\langle\T \fieldSt^{\alpha_{i_0}}(i_0)\fieldSt^{\alpha_i}(i)\rangle$.
We arrive at 
$$
T_{n+1} 
= \sum_{\alpha_0,\dots,\alpha_n} \sum_{G'} \prod_{l\in \Eint'}
\bigl< \T \fieldSt^{\alpha_{s(l)}}(s(l))
\fieldSt^{\alpha_{r(l)}}({r(l)}) \bigr> \, 
\Wickli \prod_{l\in\Eext'} \fieldSt^{\alpha_{s(l)}}({s(l)}) \Wickre
$$
where $G'$ runs through all graphs with vertices $\{0,\dots,n\}$,
internal lines $\Eint'$ and external lines $\Eext'$. Now the index
$i_0$ (which depends on the tuple $\underline{\alpha}$) is not
discriminated any more, and we can perform the sum over $\alpha$'s:
\begin{align*}
T_{n+1} 
& = \sum_{G'} \prod_{l\in \Eint'} 
\biggl< \T \sum_{\alpha_{s(l)}} \fieldSt^{\alpha_{s(l)}}({s(l)})
\sum_{\alpha_{r(l)}} \fieldSt^{\alpha_{r(l)}}(r(l)) \biggr> \,
\Wickli \prod_{l\in\Eext'} 
\sum_{\alpha_{s(l)}} \fieldSt^{\alpha_{s(l)}}(s(l)) \Wickre
\\
&= \sum_{G'} \prod_{l\in \Eint'} 
\bigl< \T \fieldSt(s(l)) \fieldSt(r(l)) \bigr> \,
\Wickli \prod_{l\in\Eext'} \fieldSt(s(l)) \Wickre \,.
\end{align*}
This is just the claimed equation~\eqref{eqTWickLinear}. 
\end{proof}

An extension of the time-ordered product across the large
string-diagonal is not yet defined up to this point. To fix it, one
first extends the Feynman propagator across~$\Delta_2$. A basic 
\mbox{(re-)}normalization condition is that the scaling degree may not
be increased. One valid extension consists in replacing
$\delta(p^2 - m^2)\theta(p_0)$ by $i/[2\pi(p^2 - m^2 + i\eps)]$ in the
Fourier transform of the two-point function. The question of whether
other extensions are permitted depends on the scaling degrees of the
Feynman propagator with respect to
the various submanifolds of~$\Delta_2$ and their respective
codimensions. We consider an example in
Appendix~\ref{app:IntegrateT}. In a second step, one can define the
time-ordered product of $n > 2$ strings
by Wick's expansion~\eqref{eqTWickLinear}. This
would amount to requiring Wick's expansion as a further normalization
condition.

\section{Final comments} 
\label{sec:Outlook}

We have constructed (outside $\LargeDiag$) the time-ordered products
of string-localized linear fields, but not of Wick polynomials. The
construction of the latter runs into the following problem. For
simplicity we consider a Wick monomial of the form
\begin{equation} 
\label{eqWickMon} 
W(x,e) \doteq \Wickli \chi(x) \fieldSt(x,e) \Wickre \,,
\end{equation}
where $\chi$ is a point-localized field with non-vanishing two-point
function with $\fieldSt$. (For example, $\chi = \fieldPt$ from
Eq.~\eqref{eqFieldSt}.) We wish to tell just from the requirements
\ref{TLinearity}, \ref{TSymmetry} and \ref{TCausality} what
$\T W(x,e) W(x',e')$ is, if the strings $S \doteq \String{x,e}$ and
$S' \doteq \String{x',e'}$ do not intersect, yet are not comparable. A
typical case is when one string, say $S$, emanates from the causal
future of $S'$ and ends up in its causal past. The best we can do is
to cut $S$ into two pieces $S = S^1 \cup S^2$ as in
Prop.~\ref{StringChopping2} such that $S^1$ is the finite segment
(containing the point $x$) and $S^1 \later S'$, while $S^2$ is the
infinite tail and $S^2 \earlier S'$. ($S^2$ is of the form 
$S^2 = x + [s_0,\infty)\,e$.) Then, as in Eq.~\eqref{eqFieldSum}, $W$
is a sum $W = W^1 + W^2$, where in particular
$$ 
W^2(x,e) = \int_{s_0}^\infty ds\, u(s)\,
\Wickli\chi(x)\,\fieldPt(x + se)\Wickre. 
$$
Now the problem is that $W^2$ is, due to the factor $\chi(x)$,
\emph{not} localized on~$S^2$ -- rather, it is ``bi-localized'' on
$\{x\} \cup S^2$! Note that $S^2$ is earlier than $S'$ but $x$ is
\emph{not}, since it is in $\overline{V_+(S')}$. Therefore, 
$\T W^2(x,e) W(x',e')$ is not fixed by \ref{TCausality}, in particular
it does not factorize as $W(x',e') W^2(x,e)$ even though 
$S' \later S^2$. Similar considerations hold for more general Wick
monomials of the form $\Wickli \chi(x)^l\,\fieldSt^k(x,e) \Wickre$.

We conclude that, in contrast to the linear case, the time-ordered
products of Wick monomials are fixed by the axioms~\ref{TLinearity}
through \ref{TCausality} only outside an open set, namely the set of
pairs of strings which are incomparable. The extension into this set
requires an infinity of parameters: It cannot be fixed by a finite set
of normalization conditions.

We conjecture that this problem can be solved as follows. Recall from
Section~\ref{sec:Intro} that in the construction of interacting models
one has to start from an interaction Lagrangian that differs from some
point-localized Lagrangian by a divergence. For the point-localized
Lagrangian $L^\pt$ there holds the strong form of Wick's expansion
outside the large (point-) diagonal, which fixes the products
$\T L^\pt \cdots L^\pt$ through the Feynman propagators. We conjecture
that the required string-independence condition (equivalence of the
string- and point-localized Lagrangians) implies that the same
expansion holds for the string-localized Lagrangian outside
$\LargeDiag$ (where it is well defined). From here, one would have to
extend the product of Feynman propagators in various steps
across~$\LargeDiag$. In models like massive QED, where the interaction
Lagrangian $j^\mu(x)\ASt_\mu(x,e)$ is linear in the string-localized
field $\ASt_\mu$, we conjecture that the SI condition fixes the
extension outside the large \emph{point} diagonal. The question of
renormalizability then amounts to the question of whether the complete
extension is fixed by a finite number of parameters, which does not
increase with the order~$n$. This is work in progress.

\appendix

\section{Extension of a string-localized Feynman propagator across
  the string-diagonal} 
\label{app:IntegrateT} 

In Prop.~\ref{T2Diag}, we have seen that the time-ordered product
$\T \field \field$ is fixed outside the string diagonal~$\Delta_2$.  
We illustrate here the extension across $\Delta_2$ with a concrete
example, which motivates that one should take the 
string-localized fields as basic building blocks even though they 
have been introduced as
integrals~\eqref{eqFieldSt} over point-localized fields.

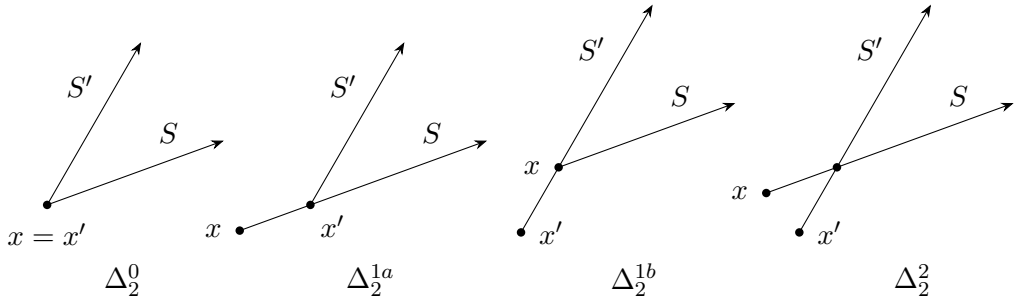
\begin{figure}[htb]
\centering
\begin{tikzpicture}[>=Stealth]
\begin{scope}
\coordinate (O) at (0,0) ; 
\coordinate (D) at (1,-1) ; 
\coordinate (E) at (20:1cm) ; \coordinate (Ep) at (60:1cm) ;
\coordinate (Xa) at ($ (O)!2.0!(E) $) ;
\coordinate (Xp) at ($ (O)!2.5!(E) $) ;
\coordinate (Ya) at ($ (O)!1.5!(Ep) $) ;
\coordinate (Yp) at ($ (O)!2.5!(Ep) $) ;
\draw[->] (O) node[below=3pt] {$x=x'$} -- (Xa) node[above left] {$S$}
   -- (Xp) ; 
\draw[->] (O) -- (Ya) node[above left] {$S'$} -- (Yp) ;
\draw (D) node {$\Delta_2^0$} ; 
\fill (O) circle (1.6pt) ;
\end{scope}
\begin{scope}[xshift=3.5cm]
\coordinate (O) at (0,0) ; 
\coordinate (D) at (0.8,-1) ; 
\coordinate (E) at (20:1cm) ; \coordinate (Ep) at (60:1cm) ;
\coordinate (X) at ($ (O)!-1.0!(E) $) ;
\coordinate (Xa) at ($ (O)!2.0!(E) $) ;
\coordinate (Xp) at ($ (O)!2.5!(E) $) ;
\coordinate (Ya) at ($ (O)!1.5!(Ep) $) ;
\coordinate (Yp) at ($ (O)!2.5!(Ep) $) ;
\draw[->] (X) node[left=3pt] {$x$} -- (Xa) node[above left] {$S$}
   -- (Xp) ; 
\draw[->] (O) node[below right] {$x'$} -- (Ya) node[above left] {$S'$}
   -- (Yp) ;
\draw (D) node {$\Delta_2^{1a}$} ; 
\foreach \pt in {O,X} \fill (\pt) circle (1.6pt) ;
\end{scope}
\begin{scope}[xshift=6.8cm, yshift=5mm]
\coordinate (O) at (0,0) ; 
\coordinate (D) at (1,-1.5) ; 
\coordinate (E) at (20:1cm) ; \coordinate (Ep) at (60:1cm) ;
\coordinate (Xa) at ($ (O)!2.0!(E) $) ;
\coordinate (Xp) at ($ (O)!2.5!(E) $) ;
\coordinate (Y) at ($ (O)!-1.0!(Ep) $) ;
\coordinate (Ya) at ($ (O)!1.5!(Ep) $) ;
\coordinate (Yp) at ($ (O)!2.5!(Ep) $) ;
\draw[->] (O) node[left=3pt] {$x$} -- (Xa) node[above left] {$S$}
   -- (Xp) ; 
\draw[->] (Y) node[right=3pt] {$x'$} -- (Ya) node[above left] {$S'$}
   -- (Yp) ;
\draw (D) node {$\Delta_2^{1b}$} ; 
\foreach \pt in {O,Y} \fill (\pt) circle (1.6pt) ;
\end{scope}
\begin{scope}[xshift=10.5cm, yshift=5mm]
\coordinate (O) at (0,0) ; 
\coordinate (D) at (1,-1.5) ; 
\coordinate (E) at (20:1cm) ; \coordinate (Ep) at (60:1cm) ;
\coordinate (X) at ($ (O)!-1.0!(E) $) ;
\coordinate (Xa) at ($ (O)!2.0!(E) $) ;
\coordinate (Xp) at ($ (O)!2.5!(E) $) ;
\coordinate (Y) at ($ (O)!-1.0!(Ep) $) ;
\coordinate (Ya) at ($ (O)!1.5!(Ep) $) ;
\coordinate (Yp) at ($ (O)!2.5!(Ep) $) ;
\draw[->] (X) node[left=3pt] {$x$} -- (Xa) node[above left] {$S$}
   -- (Xp) ; 
\draw[->] (Y) node[right=3pt] {$x'$} -- (Ya) node[above left] {$S'$}
   -- (Yp) ;
\draw (D) node {$\Delta_2^2$} ; 
\foreach \pt in {O,X,Y} \fill (\pt) circle (1.6pt) ;
\end{scope}
\end{tikzpicture} 
\caption{Configurations in submanifolds of the string diagonal $\Delta_2$}
\label{fg:crossing-lines} 
\end{figure}

The string diagonal $\Delta_2$ decomposes into the following
disjoint submanifolds: 
\begin{align*}  
\Delta_2^0 &= \set{(x,e,x',e') : x = x'} \,,
\\
\Delta_2^{1a}
&= \set{(x,e,x',e') : \exists r > 0 \text{ with } x' = x + re} ,
\\
\Delta_2^{1b} 
&= \set{(x,e,x',e') : \exists r' > 0 \text{ with } x = x' + r'e'},
\qquad   
\\ 
\Delta_2^2 &= \set{(x,e,x',e') : e,e' \text{ lin.~indep.} 
\w \exists r,r' > 0 \text{ with } x + re = x' + r'e'}.
\end{align*}
Here $\Delta_2^0$ consists of the pairs of strings with the same
initial point (the point-diagonal); $\Delta_2^{1a}$ is the set of
configurations where $x'$ lies in the relative interior of the string
$\String{x,e}$, i.e., $x' \in \String{x,e} \setminus \{x\}$; and
$\Delta_2^2$ is the set of pairs of strings whose interiors cross:
see Fig.~\ref{fg:crossing-lines}.
Thus $\Delta_2^0$ is the boundary of either $\Delta_2^{1a}$
or~$\Delta_2^{1b}$, and 
$\Delta_2^1 \doteq \Delta_2^{1a} \cup \Delta_2^{1b}$ is the boundary
of~$\Delta_2^2$.
So one must extend the Feynman propagator successively across
$\Delta_2^2$; then $\Delta_2^{1a}$ and $\Delta_2^{1b}$; and finally
across~$\Delta_2^0$.

As an example we consider massive particles of spin one and take a
line integral over the Proca field $\APt_\mu$, the so-called escort
field \cite{MundOliveira,Bert_BeyondGauge}: 
\begin{equation} 
\label{eqEscort} 
\phi(x,e) \doteq \int_0^\infty ds\, \APt_\mu(x + se)\, e^\mu.
\end{equation}
Its two-point function $\langle \phi(x,e) \phi(x',e')\rangle$ in
momentum space~\cite{MundOliveira} is
$$
\frac{1}{m^2} 
- \frac{e \cdot e'}{(p \cdot e - i\eps)(p \cdot e' + i\eps)} 
$$
times the on-shell delta distribution $\delta(p^2 - m^2) \theta(p_0)$.
It has scaling degree $0$ with respect to $\Delta_2^2$ and 
$\Delta_2^1$, and scaling degree~$2$ with respect to $\Delta_2^2$ due
to the first term.%
\footnote{These scaling degrees are calculated in~\cite{MundSantos}.}
The same holds for the Feynman propagator (outside~$\Delta_2$), and
its extension across $\Delta_2$ may not exceed this (that is the basic
renormalization condition). For all three submanifolds, the
codimension is larger than the respective scaling degree, namely $2$,
$3$ and~$4$, respectively. Therefore the respective extensions are
unique~\cite{SchulzPhD}, and the Feynman propagator is fixed without
any freedom: it is defined by replacing 
$\delta(p^2 - m^2) \theta(p_0)$ by $i/[2\pi(p^2 - m^2 + i\eps)^{-1}]$.

On the other hand, the two-point function of the Proca field in
momentum space is $\bigl( -g_{\mu\nu} + p_\mu p_\nu/m^2 \bigr)$ times
the on-shell delta distribution. Its scaling degree with respect to
the origin is~$4$; hence the Feynman propagator admits a
renormalization of the form
\begin{equation} 
\label{eqTBB} 
c\,g_{\mu\nu} \,\delta(x - x')
\end{equation}
as is well known. So, if one defines the Feynman propagator as the
integral
\begin{equation} 
\label{eqTInt} 
\int_0^\infty ds \int_0^\infty ds' \, 
\langle \T \APt_\mu(x + se) \APt_\nu(x' + s'e') \rangle \,
e^\mu e'^\nu,  
\end{equation}
as one might be tempted to do from Eq.~\eqref{eqEscort}, then by
\eqref{eqTBB} one has the freedom of adding the distribution
$$
c\,e \cdot e' \int_0^\infty ds \int_0^\infty ds'\,
\delta(x + se - x' - s'e), 
$$
supported on $\Delta_2^2$: one has an undetermined constant and
therefore has gained nothing, in contrast to the first approach where
one takes $\phi(x,e)$ as basic building block.

\section{Basic geometric notions} 
\label{app:BasicGeo}

\begin{defn} 
The forward lightcone $V_+$ is the set of all timelike and
future-pointing vectors, namely 
$V_+ \doteq \set{x \in \RR^4 : x^0 > |\xx|}$, and for $z \in \RR^4$ 
we denote
$V_+(z) \doteq V_+ + z = \set{x + z \in \RR^4 :x^0 > |\xx|}$, where
$x^0$ is the time coordinate of $x$ and $\xx$ its spatial
coordinates, i.e., $x \doteq (x^0,\xx)$.

Similarly, the backward lightcone $V_-$ is the set of all timelike and
past-pointing vectors and
$V_-(z) \doteq \set{x + z \in \RR^4 : x^0 < -|\xx|}$. 
\end{defn}

\begin{remk} 
The boundaries of the backward and forward lightcones are given by
$\partial V_-(z) \doteq \set{x + z \in \RR^4 : x_0 = -|\xx|}$ and
$\partial V_+(z) \doteq \set{x + z \in \RR^4 : x_0 = |\xx|}$,
respectively. Furthermore, the closure of the forward lightcone
$V_+(z)$ is denoted $\overline{V_+(z)}$, and similarly for the
backward lightcone.
\end{remk}

\begin{defn} 
For any set $R \subset \RR^4$ the 
causal past and future of~$R$ are defined by 
$V_-(R) \doteq \bigcup_{z\in R} V_-(z)$ and
$V_+(R) \doteq \bigcup_{z\in R} V_+(z)$, respectively.
\end{defn}

We use the following well-known fact that the forward light cone is
self-dual~\cite{ReSiII}.

\begin{lema} 
\label{VSelfdual} 
A vector $\xi \in \RR^4$ is contained in $\overline{V_+}$ if and only
if it satisfies $u \cdot \xi \geq 0$ for all $u \in V_+$.
\end{lema}

\begin{proof}
According to~\cite{ReSiII}, the open forward light cone coincides with
the interior of the set of all $\xi \in \RR^4$ that satisfy
$u \cdot \xi > 0$ for all $u\in V_+$, that is, with the interior of
the intersection
$$ 
\bigcap_{u\in V_+} \set{\xi : \xi \cdot u > 0}.
$$
But the closure of this set is just
$\set{\xi : \xi \cdot u \geq  0 \text{ for all } u \in V_+}$,
and this proves the claim. 
\end{proof}

\begin{defn} 
Given a set $A \subset \RR^4$, we define its 
\emph{causal complement} $A^c$ to be
$$
A^c \doteq \set{x \in \RR^4 : (x - y)^2 < 0 \text{ for all } y \in A}.
$$
\end{defn}

\begin{defn} 
\label{DoCo}
Let $x,y \in \RR^4$ be such that $y \in V_+(x)$. Then, we define the
open \emph{double cone} $D(y,x)$ with $x$ and $y$ as apices by
$$
D(y,x) \doteq V_+(x) \cap V_-(y).
$$
\end{defn}

The $2$-planes in Minkowski space have a classification similar to
that of the $3$-planes (see \cite{ThomasWichmann}, for instance). A
$2$-plane is \emph{spacelike} if all its nonzero vectors are
spacelike; it is \emph{timelike} if it contains a nonzero timelike
vector; it is \emph{lightlike} if it contains a line of lightlike
vectors but no timelike vectors.

Let $e,e'$ be spacelike unit vectors, that is, 
$e \cdot e = -1 = e' \cdot e'$.

\begin{lema} 
\label{Spanee'} 
\begin{enumerate}
\item[\textup{(i)}]
The linear span of $e,e'$ is timelike, lightlike or spacelike,
respectively, if and only if $(e \cdot e')^2 - 1$ is positive, zero or
negative, respectively.
\item[\textup{(ii)}]
This linear span is timelike if and only if one of the vectors 
$e \pm e'$ is timelike and the other is spacelike;
it is lightlike if and only if 
one of the vectors $e \pm e'$ is lightlike; 
it is spacelike if and only if the vectors $e \pm e'$ are either both
timelike or both spacelike.
\end{enumerate}
\end{lema}

\begin{proof}
The first statement is readily verified from the equation
$$
(se + te')^2 = - (s^2 - 2st\,e\cdot e' + t^2).
$$
Now note that $(e \mp e')^2 = -2(1 \pm e \cdot e')$, hence 
(i) implies~(ii), on account of
$$
(e \cdot e')^2 - 1 \equiv (e \cdot e' + 1)(e \cdot e' - 1)
= - \frac{1}{4}(e - e')^2 \,(e + e')^2.
\eqno \qed 
$$
\hideqed
\end{proof}

\paragraph{Acknowledgements.}
This research was generously supported by the program ``Research in
Pairs'' of the Mathematisches Forschungsinstitut at Oberwolfach in
November 2015. JM and JCV are grateful to Jos\'e~M. Gracia-Bond\'ia
for helpful discussions, at Oberwolfach and later.
We thank the referee for pertinent comments, which helped to 
fine-tune the paper.
JM and LC have received financial support by the Brasilian research
agencies CNPq and CAPES, respectively. 
They are also grateful to Finep.
The project has received funding from the European Union's
Horizon~2020 research and innovation programme under the Marie
Sk{\l}odowska-Curie grant agreement No.~690575.
JCV acknowledges support from the Vicerrector\'ia de
Investigaci\'on of the Universidad de Costa~Rica.

\providecommand{\bysame}{\leavevmode\hbox to 3em{\hrulefill}\thinspace}

\end{document}